\documentclass[letterpaper, 10 pt, conference]{ieeeconf}

\usepackage{cite}
\usepackage{amsmath,amssymb,amsfonts}
\usepackage{algorithmic}
\usepackage{graphicx}
\usepackage{textcomp}
\usepackage{xcolor}
\usepackage{bm}
\usepackage{arydshln}
\usepackage{subfigure}
\usepackage{color}
\usepackage{float}
\usepackage{url}

\newtheorem{theorem}{Theorem}

\newtheorem{definition}{Definition}

\newtheorem{example}{Example}

\def\BibTeX{{\rm B\kern-.05em{\sc i\kern-.025em b}\kern-.08em
    T\kern-.1667em\lower.7ex\hbox{E}\kern-.125emX}}
\begin{document}

\title{\LARGE \bf
Graph-Theoretic Analyses and Model Reduction for an Open Jackson Queueing Network}
\author{Chenyan Zhu and Sandip Roy}

\maketitle
\thispagestyle{empty}
\pagestyle{empty}

\begin{abstract}
A graph-theoretic analysis of the steady-state behavior of an open Jackson queueing network is developed.  In particular, a number of queueing-network performance metrics are shown to exhibit a spatial dependence on local drivers (e.g. increments to local exogenous arrival rates), wherein the impacts fall off across graph cutsets away from a target queue.  This graph-theoretic analysis is also used to motivate a structure-preserving model reduction algorithm, and an algorithm that exactly matches performance statistics of the original model is proposed.  The graph-theoretic results and model-reduction method are evaluated via simulations of an example queueing-network model.
\end{abstract}

\section{Introduction}
Queueing network models are used for communication systems analysis, emergency response planning, transportation systems design, and factory optimization among many other applications.  Motivated by these applications, a wide family of queueing-network models have been developed, which admit different types of formal analyses and are also amenable to simulation \cite{chen2001fundamentals,kelly,kleinrock}.  Among these various models, queueing networks with Poisson arrival processes and memoryless servers -- which are often referred to as Jackson networks -- are of special interest, both because they are representative of many typical queueing processes and because of their tractability \cite{chen2001fundamentals,kelly,jackson}.  Of particular note, the asymptotic joint queue-length distribution of these models have an explicitly-computable product form \cite{jackson}. A number of other statistics of the models, including queueing delays and network stay times of jobs, are also eminently tractable.  

A number of applications are requiring queueing-network models of increasingly large dimension and sophistication.  For example, the product distribution networks of online shipping merchants, which can be represented using queueing networks, are of enormous dimension \cite{retail}.  Similarly, voice and data communications backbones have grown extremely rapidly in both size and complexity \cite{kushner}.  Meanwhile, queueing models for service processes such as emergency response are sometimes becoming spatially intertwined, because of new autonomy and Internet-based coordination \cite{emergency1,emergency2}.  Because modern applications require queueing models of very high dimension, formal analyses of these models are often computationally unappealing even when the models are tractable, e.g. if they are of the Jackson network form.  At the same time, simulation-based analyses and even some formal results do not give clear insights into the underlying processes, because of their scale.  These challenges motivate new perspectives and approaches for queueing-network analysis. 

In this article, we explore two new approaches for the 
analysis of queueing networks, which are meant to support evaluation of open Jackson queueing network models of large scale:  
\begin{enumerate}
\item We build on the classical steady-state analysis of the Jackson network model, to tie queueing performance metrics (e.g., delays, network stay times) with the routing graph of the queueing network.  This graph-theoretic analysis parallels spatial analyses of linear dynamical network models (e.g. models for infection spread or synchronization), which have been recently developed within the controls community \cite{studli,kooreh,vosughi2019target}.  In the queueing-network context, the outcome of this analysis is a spatial majorization of certain performance metrics with respect to cutsets in the network's graph, which indicates that traffic interdependencies are localized.
\item  The graph-theoretic analysis of the Jackson network model provides a basis for a structure-preserving model reduction algorithm, which generates a lower-dimensional approximation to a queueing network which preserves a portion of the network's structure while abstracting away other queues.  The reduction is shown to exactly replicate some performance metrics of the original model.  The idea of structure-preserving model reduction aligns with model reduction techniques for linear differential-equation models (e.g. coherency-based reduction for the power-system models) (e.g. \cite{power1,reduction1}). However, to the best of our knowledge, this is the first effort to develop systematic model reduction techniques for queueing networks.
\end{enumerate}

The remainder of the article is organized as follows.  Section II reviews the Jackson queueing-network model and articulates the graph-theoretic analysis and model reduction goals of our effort.  Section III presents the basic graph-theoretic or spatial analysis of the model.  In Section IV, this analysis is used to characterize a number of queueing-network performance metrics.  The model reduction algorithm is developed in Section V.  Finally, 
the formal results are illustrated using an example in Section VI.

\section{Modeling and Aims}
An open Jackson network with $M$ memoryless (exponential) single-server queues, labeled $1,\hdots, M$, is considered.  Jobs are modeled as arriving at each queue $i$ from outside the network according to a Poisson process with rate $\rho_{0i}$. Jobs arriving at each queue $i$ (from outside the network or from other queues) are served according to a first-come-first-served (FCFS) discipline, at constant rate $\mu_i$; the queues are assumed to have infinite buffers.  Once a job at queue $i$ completes service, it is either routed to queue $j$ ($j=1,\hdots, m$) with probability $r_{ij}$, or departs the system with probability $r_{i0}$, where $\sum_{j=0}^m r_{ij}=1$.
A weighted directed graph (digraph) $\Gamma$ with $M$ vertices (labeled $i=1,\hdots, M$) is used to represent the routing protocol of the queueing network.  Specifically, each queue $i$ in the network is represented by the corresponding vertex $i$ in the digraph. An edge is drawn from vertex $i$ to vertex $j$ if $r_{ij}>0$, and the edge is assigned a weight of $r_{ij}$. The digraph associated with a queueing network is illustrated in Figure \ref{fig1}.  We also define the $M \times M$ routing matrix $R$ as $R=[r_{ij}]$.

\begin{figure}[thpb]
\centerline{\includegraphics[scale=0.5]{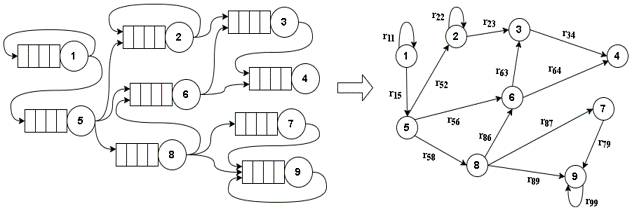}}
\caption{An illustration of a Jackson queueing system with 9 queues modeling by a graph $\Gamma$.}
\label{fig1}
\end{figure}

Analyses of queueing networks are often primarily focused on: 1) queue-length statistics and distributions; 2) delay- and sojourn- times of jobs at individual queues as well as overall stay times in the network; and 3) utilization fractions of individual queues. For the open Jackson network model, mean values of these performance metrics can be found by first finding the equilibrium total arrival rates into each queue $i$, which we denote as $\lambda_i$ \cite{chen2001fundamentals}.  Specifically, the arrival rates can be found by solving
the following linear system of equations: $\lambda_i=\sum_{j=1}^{j=M} \lambda_jr_{ji}+
    \rho_{0i}$, $i=1,\hdots, M $, provided that the
network reaches a statistical equilibrium which happens if $\lambda_i < \mu_i$ for all $i$.  We assume from here on that this condition holds. These equations, which are known as the traffic equations, can be rewritten in vector form as:$ \bm{\lambda}={{R}'\bm{\lambda}}+\bm{\rho}$, where $\bm{\lambda}=\begin{bmatrix} \lambda_1 \\ \vdots \\ \lambda_M \end{bmatrix}$, $\bm{\rho}=\begin{bmatrix} \rho_{01} \\ \vdots \\ \rho_{0M} \end{bmatrix}$, and $R'$ is the transpose of the matrix $R$.  
The solution to the traffic equation is \cite{chen2001fundamentals} 
\begin{equation}
    \bm{\lambda}=(I-R')^{-1}\bm{\rho} \label{1} .
\end{equation}

Once the equilibrium arrival rates have been determined, mean values of other performance metrics can readily be computed.  
Specifically, the mean queue length (not including the job that is being served) is given by \cite{agrawal2015introduction} 
\begin{equation}
    L_i=E[q_i(t)-1]=\lambda_i^2/(\mu_i^2-\mu_i\lambda_i) .   \label{2}
\end{equation}
The mean sojourn time for a job in queue $i$ (the mean time spent by a job in the queue) is given by \cite{queuedelay2005}:
\begin{equation}
    W_i=1/(\mu_i-\lambda_i)  . \label{3} \\
\end{equation}
The mean delay for a job in queue $i$ (the mean time spent by the job while waiting for service) is therefore given
by \cite{queuedelay2005}:
\begin{equation}
    D_i=1/(\mu_i-\lambda_i)-\frac{1}{\mu_i}  . \label{4} \\
\end{equation}
Additionally, the mean overall time in the network of a job originally entering the network at queue $i$ (i.e. the stay time of the job) can be found as \cite{queuedelay2005}:
\begin{equation}
    {T_{i,d}}=[(I-R)^{-1}\bm{W}]_i,  \label{5}
\end{equation}
where $\bm{W}=\begin{bmatrix} W_1 \\ \vdots \\ W_M \end{bmatrix}$.

Relevant to our development, the solution to the traffic equation given in $\eqref{1}$ allows a concise interpretation for the contribution of exogenous arrivals to the queueing process at an individual queue.  Specifically, let us consider a row $i$ of the matrix $(I-R')^{-1}$.  
The entries in the row define how exogenous arrivals at each queue eventually contribute to the total arrival rate (internal or external) at queue $i$.  Thus, they decide how exogenous traffic increments contribute to the total arrival demand at the queue, i.e. they show the relative influences of different exogenous traffic streams on traffic demand (and hence queue length and delay) at a particular queue.
Based on this interpretation, we label the matrix $(I-R')^{-1}$ as the contribution matrix, and denote it as ${\alpha}$. Specifically, each entries of the contribution matrix denoted as $\alpha_{ij}$ is interpreted as a scaling factor of exogenous arrivals at queue $j$ contributing to demand at queue $i$. 

The research described here pursues two main aims with regard to the open Jackson network model.  The first aim is to understand whether and how arrival demands, and hence queueing performance metrics, depend on the graph topology of the Jackson network model.   Conceptually, it might be expected that total arrival demand at a particular queue is most strongly influenced by exogenous arrivals at nearby queues, i.e. nearby exogenous arrivals have larger contributions to the total arrival rate while far-away exogenous inputs necessarily have limited influence.  Here, we formally verify this spatial dependence of queue arrival demands on exogenous arrivals, and study implications on performance metrics (including queue length, delay, and overall stay time).  The results are important because they give an understanding of traffic propagation and disturbance impacts in queueing network models, and in particular show that queueing network dynamics exhibit a spatial localization which is relevant for analysis and design of large networks.

The second aim of our study is to develop and characterize a structure-preserving model reduction technique for the open Jackson network model, so as to simplify statistical analysis of the network model.  Our model reduction approach is based on the recognition that traffic propagation in queueing networks exhibits a localization.  Given this localization, it is reasonable to develop a reduction which entirely preserves a study area of interest (e.g. an area where exogenous traffic inputs may change, congestion/delays are of concern, or utilization is of interest), and condenses away the rest of the network.  This sort of structure-preserving reduction is appealing as it allows analysis of performance statistics and queue dynamics for the area of interest, with much reduced computation.  In this study, we aim to develop a model reduction strategy of this sort, and to understand how well the reduced model replicates the behavior of the original.

\section{Basic Graph Theoretic Analysis}

The focus of this section is to characterize the contribution weights in the matrix $\alpha$ in terms of the network's graph topology. The analysis requires several definitions related to the graph $\Gamma$ of the queueing network model, presented next:

1) One vertex in $\Gamma$ (respectively queue in the Jackson network) is labeled as the target vertex (target queue).  The analysis is focused on the contribution weights and later performance statistics of the target queue.  Another vertex (queue) is labeled as the receiver vertex (queue); the terminology `receiver' is used to indicate that we are interested in the propagation/impact of exogenous arrivals received at the queue.

(2) The vertices in the graph (queues in the network) are partitioned into three sets,  a head ($H$), a cutset ($C$), and tail ($T$).  The cutset is a vertex-cutset which separates the head and tail, i.e. every path in $\Gamma$ from a vertex in the head to one in the tail passes through the cutset.  The head is always defined to contain the target queue.  Receiver queues in the different sets are considered in the development. For ease of presentation, the vertices (queues) are indexed starting with those in the Head, followed by those in the cutset, and finally those in the tail. Figure \ref{fig2} illustrates the graph concepts.

\begin{figure}[thpb]
\centerline{\includegraphics[scale=0.5]{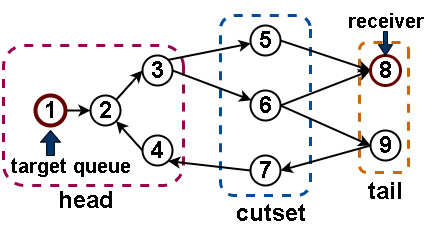}}
\caption{Network graph $\Gamma$  with vertex (queue) 1 as target and vertex (queue) 10 as receiver, where $H=\{1,2,3,4\}$, $C=\{5,6,7\}$ and $T=\{8,9\}$.}
\label{fig2}
\end{figure}
A routing matrix with a block-partition form following the above graph topology is considered:
\begin{equation}
    R=
    \begin{bmatrix}
	R_{HH} & R_{HC} & \bm{0} \\
	R_{CH} & R_{CC} & R_{CT} \\
	\bm{0} & R_{TC} & R_{TT} \\
    \end{bmatrix}\label{6}
\end{equation}
where each block contains routing probabilities (edge weights) from vertices in the first subscripted subset to vertices in the second subscripted subset.   The block-partition form is used as a default for all routing matrix analyses below.

The main result presented in this section is a  comparison of contribution weights for different receiver locations. Specifically, entries in a row of a contribution matrix, which identify the contribution weights for different receivers for a specific target, are characterized in the following theorem:

\begin{theorem}
Consider an open Jackson network with routing matrix $R$ and corresponding network graph $\Gamma(R)$. Also, consider a target $q$, a cutset $C$,  and a receiver which is located in the tail $T$ or cutset $C$.  Then the contribution weights satisfy: 
$\max_{c \in C} \alpha_{qc} \ge \alpha_{qt}$ for all $t \in T$, i.e.
the contribution for at least one receiver on the cutset majorizes the contribution for any receiver in the tail.  
\end{theorem}
\begin{proof}
 Without loss of generality, consider queue 1 as the target queue. Also assume $C$ is a $\hat{c}$ queues cutset and a tail $T$ with $\hat{t}$ queues. 
 Then the contribution weights of interest are on the first row of the contribution matrix. It is convenient to partition this vector
    as $\alpha_1^T=[\bm{x}_H^T|\bm{x}_C^T|\bm{x}_T^T]$, where the partitions correspond to the head (which includes the target vertex), the cutset, and the tail (which contains the receiver).
    From the fact that the (column) vector 
    $\alpha_1=\begin{bmatrix} \bm{x}_H \\ \bm{x}_C \\
    \bm{x}_T \end{bmatrix}$
    is the first column of $(I-R)^{-1}$, it follows
    that $(I-R) \bm{x}=\bm{e}_1$, where $\bm{e}_1$
    is an indicator vector with $1$th entry equal to 
    $1$.  
    
    Substituting the block forms, we get
    \begin{equation}
        \begin{bmatrix}
	I-R_{HH} & -R_{HC} & \bm{0} \\
	-R_{CH} & I-R_{CC} & -R_{CT} \\
	\bm{0} & -R_{TC} & I-R_{TT} \\
    \end{bmatrix} 
    \begin{bmatrix} \bm{x}_H \\ \bm{x}_C \\
    \bm{x}_T \end{bmatrix}=\bm{e}_1  . \label{7}
    \end{equation}
    From the lowest block in the equation, it immediately follows that:
    \begin{equation}
     (I-R_{TT})\bm{x}_T=R_{TC}\bm{x}_C, \label{8}
     \end{equation}
     The solution to the system of equations is 
     \begin{equation}
     \bm{x}_T=\sum_{n=0}^{\infty}R_{TT}^{n}R_{TC}\bm{x}_C, \label{9}
     \end{equation}
     provided that the summation converges.
     However, since $R_{TT}$ is the principal submatrix of a substochastic matrix, it follows that its eigenvalues are strictly within the unit circle, and the sum converges. From here, it follows that the $i$th entry of
 $\bm{x}_T$ is given by
     \begin{equation}
        [\bm{x}_T]_i=\sum_{n=0}^{\infty}[R_{TT}^{n}R_{TC}]_i \bm{x}_C . \label{10}
    \end{equation}
  Let us denote the largest entry of the vector $\bm{x}_C$ as $[x_C]_{max}$. Noting 
  that $[R_{TT}^{n}R_{TC}]_i$ is nonnegative, the following upper bound is obtained:
      \begin{equation}
        [\bm{x}_T]_i\leq\sum_{n=0}^{\infty}[R_{TT}^{n}R_{TC}]_i \begin{bmatrix} [x_C]_{max} \\ \vdots \\  [x_C]_{max}\end{bmatrix}. \label{11} 
    \end{equation}
 This can be further simplified to:
     \begin{equation}
        [\bm{x}_T]_i\leq[x_C]_{max}\sum_{j=1}^{j=\hat{c}}\sum_{n=0}^{n=\infty}  [R_{TT}^{n}R_{TC}]_{ij}. \label{12}
    \end{equation}
 Next, we recognize that $\sum_{j=1}^{j=\hat{c}}\sum_{n=0}^{n=\infty}  [R_{TT}^{n}R_{TC}]_{ij}$ represents the sum of the $i$th row of $\sum_{n=0}^{n=\infty} R_{TT}^{n}R_{TC}$. However, it can be shown that the sum of each row of $\sum_{n=0}^{n=\infty} R_{TT}^{n}R_{TC}$ is less than 1. To show this, consider the matrix $M_R=\begin{bmatrix} R_{TT} & R_{TC} \\ \bm{0} & I_{\hat{c}} \\ \end{bmatrix}$, the sum of each of the first $\hat{t}$ rows of $M_R$ is less than or equal 1 since $R_{TC}$ and $R_{TT}$ are submatrices of $R$. It follows that the each of the first $\hat{t}$ rows of $M_R^n=\begin{bmatrix} R_{TT}^n & \sum_{k=0}^{n-1}R_{TT}^n R_{TC} \\ \bm{0} & I_{\hat{c}} \\ \end{bmatrix}$ has sum less than or equal to 1, for any $n=0,1,\hdots$. The result therefore holds in the limit, i.e. each row of $\sum_{n=0}^{n=\infty} R_{TT}^{n}R_{TC}$ is less than 1.  Therefore, $\sum_{j=1}^{j=\hat{c}}\sum_{n=0}^{n=\infty}  [R_{TT}^{n}R_{TC}]_{ij}\leq 1$. Hence, the upper-bound given by \eqref{12} can be expressed to:
      \begin{equation}
        [\bm{x}_T]_i\leq[x_C]_{max}. \label{13}
    \end{equation}
 Recall that $[x_C]_{max}$ is the largest-magnitude contribution weight for a cutset queue, and $[\bm{x}_T]_i$ represents the contribution weight of any queue in the tail (including the selected receiver). Therefore, the theorem follows.  
\end{proof}

Theorem 1 shows that contribution weights for a selected target queue show a spatial dependence. Specifically, a particular receiver queue has a smaller contribution as compared to at least on receiver queue on a separating cutset. An immediate consequence is that the maximum contribution weight among receivers at a certain distance from the target in the network graph is a non-increasing function of the distance.

In general, although maximum contribution weights on network cutsets show a spatial dependence, the exact pattern of contribution weights is complicated.  However, for some special cases, exact analyses of the contribution weights are possible.  In the following theorem, we exactly compute
the contribution weights in the case that the network graph is a (bidirectional) line graph, i.e. the routing matrix is Toeplitz \cite{anguluri2022network} (see Figure \ref{fig3}).

\begin{figure}[thpb]
\centerline{\includegraphics[scale=0.45]{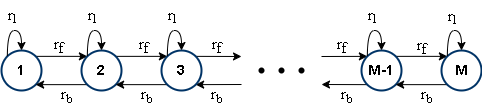}}
\caption{Line queueing network with routing probabilities $r_f$, $r_b$ and $r_l$}
\label{fig3}
\end{figure}

\begin{theorem}
Consider an open Jackson queueing network with $M$ queues, which has the linear network graph topology shown in Figure \ref{fig3}. Queue $1$ is selected as a target queue. The contribution weight when each queue $i$ is selected as the receiver is given by:
\begin{equation}
    	\alpha_{1i}=\dfrac{r_b^{i-1}(\gamma_2^{M-i+1}-\gamma_1^{M-i+1})}{r_f^i(\gamma_2^{M+1}-\gamma_1^{M+1})}, \label{14}
\end{equation}
where 
\begin{gather}
    \gamma_1= \dfrac{(1-r_l)+\sqrt{(1-r_l)^2-4r_br_f}}{2r_f}\label{15}, \\
    \gamma_2= \dfrac{(1-r_l)-\sqrt{(1-r_l)^2-4r_br_f}}{2r_f}. \label{16}
\end{gather}
\end{theorem}

\begin{proof}
    Each entry $\alpha_{1i}$ in the first row of the contribution matrix is calculated by iteration. To do so, we recall that the transpose of the first row can be found by solving $(I-R)(\alpha_1)'=\bm{e}_1$. The equations can be written out as follows:
    \begin{equation}
        (I-R)\alpha_1=\begin{bmatrix}
        (1-r_l)\alpha_{11}-r_f\alpha_{12} \\
        -r_b\alpha_{11}+(1-r_l)\alpha_{12}-r_f\alpha_{13}\\
        \vdots \\
        -r_b\alpha_{1,M-1}+(1-r_l)\alpha_{1M}
        \end{bmatrix}=\bm{e}_1. \label{17}
    \end{equation}
    
    From the expanded form of the equation, one sees that
    each $\alpha_{1j}$ is a weighted linear combination of 
    $\alpha_{i,j-1}$ and $\alpha_{i,j-2}$, with fixed weights.  Thus, it is seen that the $\alpha_{1j}$ can
    be found by solving a second-order linear difference equation, where there is an additional constraint on the
    first and last terms $\alpha_{1M}$ and $\alpha_{11}$ 
    arising from the first and last equations.  The solution
    to the equation is standard and is commonly found in  solving equations with tri-diagonal Toeplitz matrices, hence we omit the details.  The result is as given in the theorem statement.
    
\end{proof}

From Theorem 2, it is easy to verify that the contribution weights fall off roughly exponentially with the distance of the receiver from the target in the line-network case.

\section{Analyses of Performance Metrics}

The spatial analysis of contribution weights developed in Section III directly allows for graph-theoretic characterization of a range of performance metrics of interest for the Jackson network model.  

An immediate consequence of the contribution-weights analysis is that 
the impacts of exogenous arrival rates on total arrival rates to a queue also exhibit a spatial pattern.  Specifically, let us consider the differential change in total arrival rate at a target queue due to a change in the exogenous input rate
at a receiver.  This impact can be seen to be smaller than the maximum possible impact if the exogenous input rate at a queue on a separating cutset was instead modified; in other words, exogenous arrivals to nearby queues have a greater impact on the traffic density at a target queue, as compared to those that are further away.  This notion is formalized in the following theorem.

\begin{theorem}
Consider an open Jackson queueing network with a routing matrix $R$ and corresponding graph $\Gamma(R)$. Consider a particular target queue $q$, a cutset $C$, a tail $T$, and receiver locations $r$ in the cutset or the tail. Assume that the exogenous arrival rate at the receiver is increased by an increment $\delta>0$, and in consequence the arrival rate at the target queue is changed to $\lambda_q(\delta_r)$.  Then the arrival rate satisfies: $max_{r \in C} \lambda_q(\delta_r)\geq \lambda_q(\delta_r)$ for any $r \in T$, i.e. the maximum arrival rate change due to an exogenous input increment at a cutset receiver queue majorizes the change when the increment is at a tail receiver queue.
\end{theorem}
\begin{proof}
Without loss of generality, consider queue $1$ as the target queue. The mean arrival rate at queue $1$ can be written as $\lambda_1=\sum_{j=1}^{M}\alpha_{1j}\rho_{0j}$. Now consider a receiver $r$ located in the cutset which has exogenous arrivals increment $\delta$. Then the arrival rate becomes $\lambda_1(\delta_{r})=\sum_{j=1}^{M}\alpha_{1j}\rho_{0j}+\alpha_{1r}\delta$, i.e. $\lambda_1+\alpha_{1r}\delta$, where $r \in C$. Identically, when the increment is at a vertex $r \in T$,  the arrival rate at target queue is given by: $\lambda_1(\delta_{r})=\lambda_1+\alpha_{1r}\delta$.  Since $\max_{c \in C} \alpha_{ic} \ge \alpha_{1r}$ for any $r \in T$ from Theorem 1,  it follows that $\max_{r \in C} \lambda_1(\delta_r) \ge \lambda_1(\delta_r)$ for any $r \in T$.
\end{proof}

Theorem 3 shows that the exogenous arrival increments in nearby receivers have a stronger impact on the target traffic density, which reflects the fact that jobs from nearby queues are more likely to be routed to the target. The spatial result for target queue arrival rates also translates to spatial results on queue performance metrics (queue lengths, delays).

\begin{theorem}
Consider an open Jackson queueing network with a routing matrix $R$ and graph $\Gamma(R)$. Consider a particular target queue $q$, and receiver locations $r$ in the cutset $C$ or the tail $T$. Assume that the exogenous arrival rate at a receiver $r$ is increased by an increment $\delta$. The performance metrics at the target queue after the increment are denoted as $L_q(\delta_{r})$ (mean queue length), $W_q(\delta_{r})$ (sojourn time) and $D_q(\delta_{r})$ (delay). These performance metrics satisfy the following inequalities:\\ 
(1). $max_{r \in C} W_q(\delta_{r})\geq W_q(\delta_{r})$ for any $r \in T$ \\
(2). $max_{r \in C} D_q(\delta_{r})\geq D_q(\delta_{r})$ for any $r \in T$ \\
(3). $max_{r \in C} L_q(\delta_{r})\geq L_q(\delta_{r})$ for any $r \in T$ 
\end{theorem}
\begin{proof}
    Without loss of generality, consider queue 1 as the target queue. We first show the result for the mean sojourn time.  We denote the arrival rate at the target queue when the exogenous traffic increment is present as $\lambda_1(\delta_{r})$. Then the corresponding mean queue length can be written as $W_1(\delta_{r})=\dfrac{1}{\mu_1-\lambda_1(\delta_{r})}$. From Theorem 3, $max_{r \in C}\lambda_1(\delta_{r}) \ge \lambda_1(\delta_{r})$ for any $r \in T$ is given.  This yields that 
    $max_{r \in C}W_1(\delta_{r})=\dfrac{1}{\mu_1-max_{r \in C}\lambda_1(\delta_{r})}$ is no smaller than $W_1(\delta_{r})=\dfrac{1}{\mu_1-\lambda_1(\delta_{r})}$ for any $r \in T$. The results for the mean delay and queue length can be proved following the same method, hence we omit the details .
\end{proof}
For a particular target queue, perturbations to exogenous arrivals at nearby queues have larger impact on performance statistics as compared to perturbations at remote locations; thus, the main queue-specific 

The mean overall stay time for jobs entering at a target queue is also of interest.  Next, we characterize how this stay time depends on increments in the service rate at different queues.  Intuitively, the stay time is expected to depend more strongly on service-rate changes at nearer locations in the network graph. This is formalized in the following theorem: 
\begin{theorem}
Consider an open Jackson queueing network with a routing matrix $R$ and graph $\Gamma(R)$. Consider a target queue $q$, a cutset $C$, and tail $T$.  Say that there is an  increment in the service rate at the target queue, and consider the impact on the stay time for a job that enters a queue $i$ which may be in the cutset or the tail.  Then the changes in the stay times $\Delta(T_{i,d})$  (i.e. difference between the stay time before the increment and after the increment) satisfy 
$\max_{i \in C} \Delta(T_{i,d}) \ge \Delta(T_{i,d})$ for any $i \in T$. 
\end{theorem}
The proof is similar to those of the previous results, hence is omitted.
\section{Model Reduction}

We explore reduction of a Jackson network model to a smaller-dimensional model (i.e. one with fewer queues), with the goal of enabling faster simulation and analysis of the model.  The spatial analysis of the Jackson network developed in the prior sections indicates that the model has a localization property, in the sense that traffic flows are predominantly impacted by nearby exogenous inputs and queueing dynamics.  This suggests that model reductions which preserve a portion of a Jackson network while aggregating or reducing away other queues in the network may give good approximations.  Here, we develop a model reduction strategy of this type, with the goal of preserving the contribution weights (and hence resultant performance statistics) in a portion of the queueing network.  

Prior to presenting the model-reduction algorithm, we first define a notion of an {\em equivalent reduction} for a Jackson network:

\begin{definition}
Consider an open Jackson network model with $M$ queues which has routing matrix $R$ and network graph $\Gamma(R)$.  Also, consider a second open Jackson network with $D<M$ queues which has routing matrix $R^{\star}$ and network graph $\Gamma(R^{\star})$.  The second Jackson network model is said to be an {\em equivalent reduction} of the first model if  $\alpha^{\star}_{i,j}=\alpha_{i,j}$ for $i=1,2,\dots,D$ and $j=1,2,\dots,D$, where $\alpha^{\star}$ and $\alpha$ are the contribution matrices for the two models. 
\end{definition}

The following main theorem gives an algorithm for constructing an {\em equivalent reduction} for a Jackson network model.  Specifically, let us consider a subdivision of a Jackson network's queues into three sets -- a head, a cutset, and a tail -- as introduced in Section II.  A reduction for the model is sought, which reduces away the queues in the tail while preserving queues in the head and cutset.  The following theorem shows how to construct the routing matrix for the reduced model, so it is equivalent:

\begin{theorem}
Consider an open Jackson queueing network with routing matrix $R$, which
is block partitioned as given in Equation (6). An open Jackson network with the following routing matrix $R^{\star}$ is an {\em equivalent reduction} of the original model:
\begin{equation}
    R^{\star}=\begin{bmatrix}
    R_{HH} & R_{HC} \\
    R_{CH} & R_{CC}+L_C
    \end{bmatrix}, \label{18}
\end{equation}
where 
\begin{equation}
    L_C=R_{CT}(I-R_{TT})^{-1}R_{TC}. \label{19}
\end{equation}
\end{theorem}

\begin{proof}
     Consider a re-partitioning of the routing matrix $R$ as $R=\begin{bmatrix}
    R_{DD} & R_{DT} \\
    R_{TD} & R_{TT}
    \end{bmatrix}$,
    where $R_{DD}=\begin{bmatrix}
    R_{HH} & R_{HC} \\
    R_{CH} & R_{CC}
    \end{bmatrix}$
    We substitute this block matrix form into the expression $\alpha=(I-R')^{-1}$ and apply the block-partition inverse formula. This yields the following expression for the $D \times D$ submatrix $[\alpha_{i,j}]$, where $i=1,2,\dots,D$ and $j=1,2,\dots,D$:
    \begin{equation}
        [\alpha_{i,j}]=((I-(R_{DD}+R_{DT}(I-R_{TT})^{-1}R_{TD})')^{-1}. \label{20}
    \end{equation}
  
   Next, let us consider the matrix $\alpha^{\star}$, which is given by $\alpha^{\star}=(I-R^{\star '})^{-1}$.  Therefore, if we can
   show that $R^{\star }$ as defined in the theorem statement equals $R_{DD}+R_{DT}(I-R_{TT})^{-1}R_{TD})$, then the theorem statement is proved.  Noticing that $R_{DT}=\begin{bmatrix} \bm{0} \\ R_{CT} \end{bmatrix}$ and $R_{TD}=\begin{bmatrix} \bm{0} & R_{TC} \end{bmatrix}$ and substituting the form for $R_{DD}$, we see that  this equivalence holds.
\end{proof}
Theorem 6 gives a method for constructing a reduced open Jackson network model which preserves part of the original network model, while eliminating the remainder of the network. Specifically, the theorem provides an algorithm for constructing a reduced routing matrix.  Additionally, if exogenous input rates and queue service networks within the preserved subnetwork are maintained, then the reduced model has equivalent performance statistics as compared to the original model for the preserved part of the network.  The reduction has the further benefit that a portion of the routing matrix -- that corresponding to the head -- is exactly maintained, with only the routing probabilities associated with the cutset being modified.  An illustration of the reduction is given in Figure \ref{fig4}.

\begin{figure}[thpb]
\centering  
\subfigure[]{
\begin{minipage}{3.8cm}
\centering 
\includegraphics[scale=0.28]{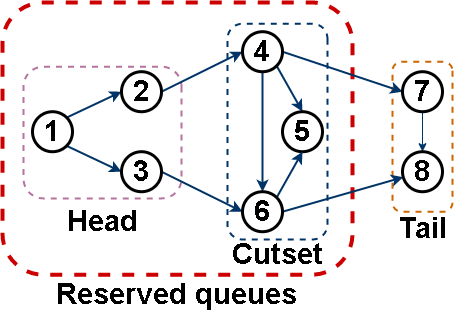}  
\end{minipage}
}
\hfill
\subfigure[]{ 
\begin{minipage}{3.8cm}
\centering    
\includegraphics[scale=0.27]{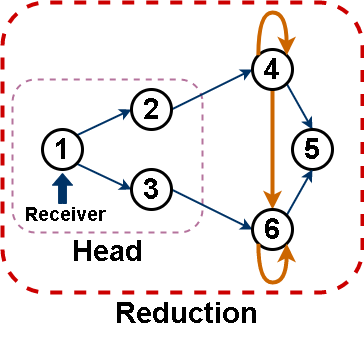}
\end{minipage}}
\caption{Illustration of model reduction. (a) An open Jackson queueing system with 8 nodes. Nodes of a tail labeled 7 and 8. (b) The reduction generates loops for a cutset nodes 4 and 5 that were connected to a tail in (a). In addition, the routing probabilities in a cutset may increase as well.}   
\label{fig4}    
\end{figure}

\begin{example}
A simple example is presented to illustrate the model reduction. An open Jackson queueing network with 5 queues is considered, which has the following routing matrix.  A {\em reduction} which maintains three queues is sought.
\begin{equation}
    R=\left[\begin{array}{cc:c:cc}
    0.4 &0.45 &0 &0	&0	 	   \\
	0.23 &0.3 &0.23 &0 &0	 	\\
	\hdashline
	0 &0.05 &0.08 &0.23 &0.15\\
    \hdashline
    0 &0 &0.21 &0.01 &0.17   \\
	0 &0 &0.29 &0.25 &0.16   
\end{array}\right] \notag
\end{equation}

The {\em equivalent reduction} given by Theorem 4 has the following routing matrix:
\begin{equation}
    R^{\star}=\left[\begin{array}{cc:c}
    0.4 & 0.45 &0 \\
	0.23 & 0.3 & 0.23\\
	\hdashline
	0 & 0.05 & 0.2103 \notag
\end{array}\right] \notag
\end{equation}
As claimed in the theorem, the reduction only changes one block in routing matrix, which corresponds to a single-vertex (single-queue) cutset. As shown below, the contribution matrix for the reduced model $\bm{\alpha}^{\star}$ is identical to the top left $3 \times 3$ block of the contribution matrix $\bm{\alpha}$ for the original model. 
\begin{equation}
    \alpha=\left[\begin{array}{ccc:cc}
	2.22 &0.75 &0.05 &0.01	&0.02 \\
	1.46 &1.95 &0.12 &0.04	&0.05 \\
	0.43   &0.57 &1.30	  &0.37  &0.56  \\
	\hdashline
	0.12 &0.17 &0.38&1.17&0.48 \\
	0.10 &0.14 &0.31 &0.30 &1.39\\
\end{array}\right] \notag
\end{equation}

\begin{equation}
    \alpha^{\star}=\begin{bmatrix}
    2.22 & 0.75 &0.05 \\
	1.46 & 1.95 & 0.12\\
	0.43 & 0.57 & 1.30  \notag
    \end{bmatrix}
\end{equation}
Thus, average performance metrics of individual queues, such as the average queue length and delay time, are also preserved in the network. 
\end{example}	

\section{Simulations}

Simulations of an open Jackson queueing network and a reduced-order approximation are undertaken, to illustrate the graph-theoretic results and evaluate the effectiveness of the model reduction. An open queueing network with 12 queues is considered. The graph $\Gamma(R)$ of routing probabilities is shown in Figure \ref{fig5}.  The service rate $\mu_i$ for each queue is assumed to be $\mu_i=20$ jobs per unit time.  To simplify the presentation,  we consider performance metrics for a target queue (Queue 1), in the case where there are exogenous arrivals at a single receiver queue at a rate of 5 jobs per unit time; this simplified case captures the impacts of individual queues' exogenous arrival processes on the target queue's performance statistics.  Figure \ref{fig6} shows a single simulation of the queue length (number of jobs) in the target queue as a function of time, under the nominal arrival regime. 

\begin{figure}[thpb]
\centerline{\includegraphics[scale=0.35]{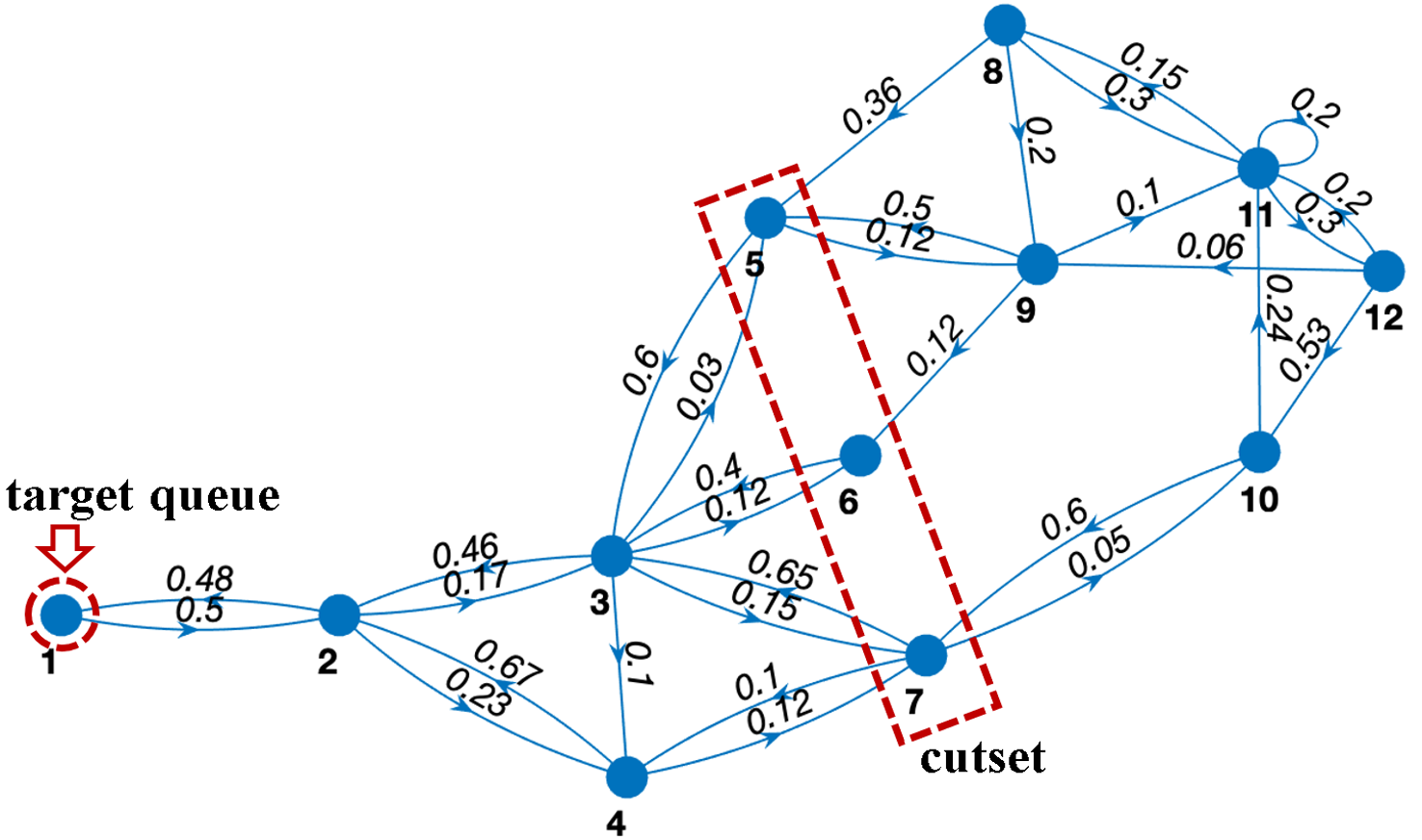}}\vspace{-0.2cm} 
\caption{Graph $\Gamma(R)$ of an open Jackson queueing network mod $\Gamma(R)$.  We study the performance when  queue 1 is selected as a target queue. } 
\label{fig5}
\vspace{-0.8cm} 
\end{figure}

\begin{figure}[thpb]
\centerline{\includegraphics[scale=0.4]{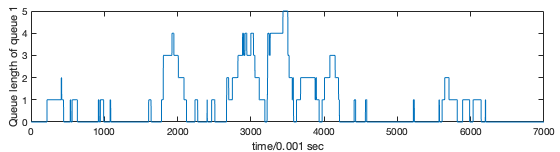}}\vspace{-0.2cm} 
\caption{The queue length of the target queue $1$ versus time.} 
\label{fig6}
\vspace{-0.2cm}
\end{figure}

Spatial patterns in queueing-network performance metrics are illustrated in Figure \ref{fig7}.  Specifically, the change in the mean queue length and mean sojourn time are determined for the target queue (queue $1$), when exogenous arrival rates are incremented at each queue (i.e. each queue is considered as a receiver with an incremented arrival rate). These changes are plotted against the receiver location in Figure \ref{fig7}.  As expected, the impacts on the target queue are largest when the receiver is close to the target, with decreasing impact for remote receivers.  As a specific example, consider the cutset comprising vertices (queues) 5, 6, and 7.  It is seen that the queue length and sojourn time impacts for one receiver queue on the cutset is larger than for any queue that is separated from the target by the cutset.

\begin{figure}[thpb]
\centering 
\subfigure[The dependence of the mean queue length at the target queue on the location of the receiver queue with an increment in jobs is shown.]{
\begin{minipage}[c]{1\linewidth}
\centering 
\includegraphics[scale=0.2]{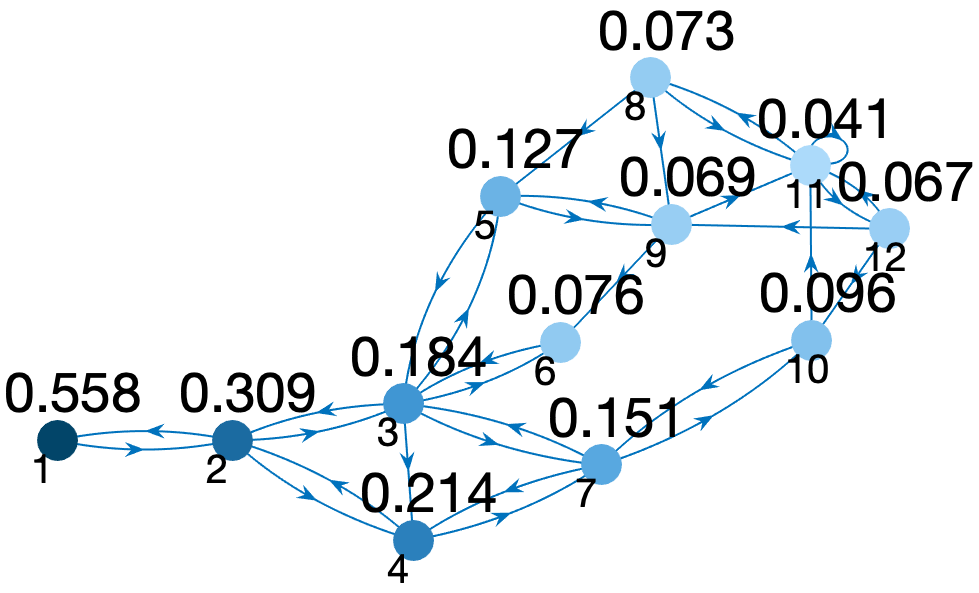}\vspace{0.2cm}
\end{minipage}
}
\hfill
\subfigure[The dependence of the mean queue length at the target queue on the location of the receiver queue with an increment in jobs is shown.]{ 
\begin{minipage}[c]{1\linewidth}
\centering    
\includegraphics[scale=0.2]{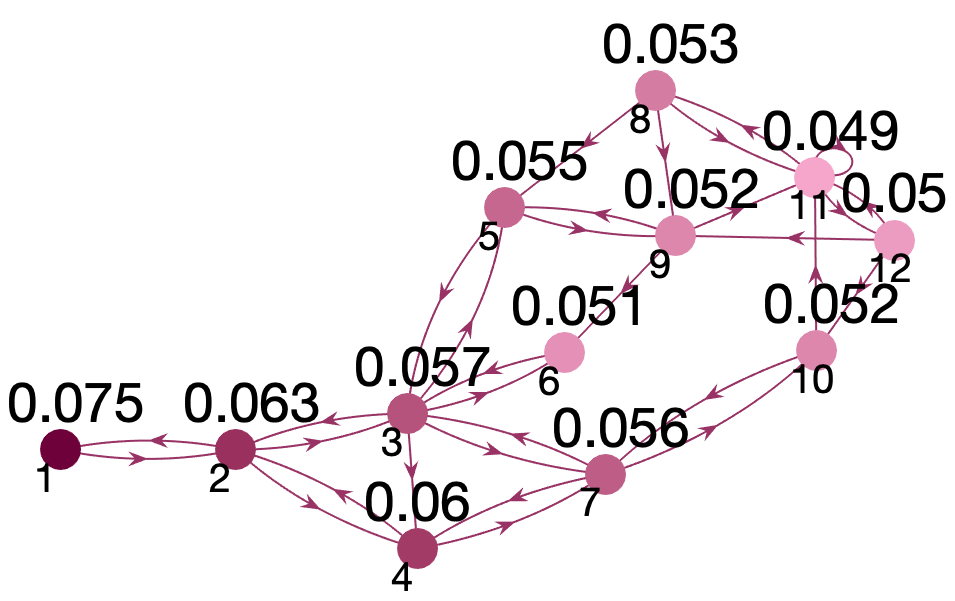}\vspace{0.2cm}
\end{minipage}}
\caption{Visualization of the performance metrics for target queue 1, for different target locations. The value shown above each queue represents the performance metric of the target queue $1$ when the corresponding queue is selected as receiver.}  
\label{fig7}
\vspace{-0.2cm}
\end{figure}

Next, a reduction of the twelve-queue network is constructed, using the procedure described in Section V.  Specifically, an {\em equivalent reduction}  which preserves Queues 1-7 is considered. The routing matrix for the reduced model is compared with the original in Figure \ref{fig8}.  As per the model reduction procedure, the network (routing) graph for Queues 1-7 is preserved, but the routing probabilities within the cutset separating the preserved and reduced parts of the network (Vertices 5-7 in this case) are modified.  A simulation of the queue length of the target queue $1$ is shown for the reduced model(Figure \ref{fig9}); the queue dynamics are similar to those of the original model.

\begin{figure}[thpb]
\centering  
\subfigure[The network graph $\Gamma(R)$ of the original queueing network.]{
\begin{minipage}[c]{1\linewidth}
\centering 
\includegraphics[scale=0.35]{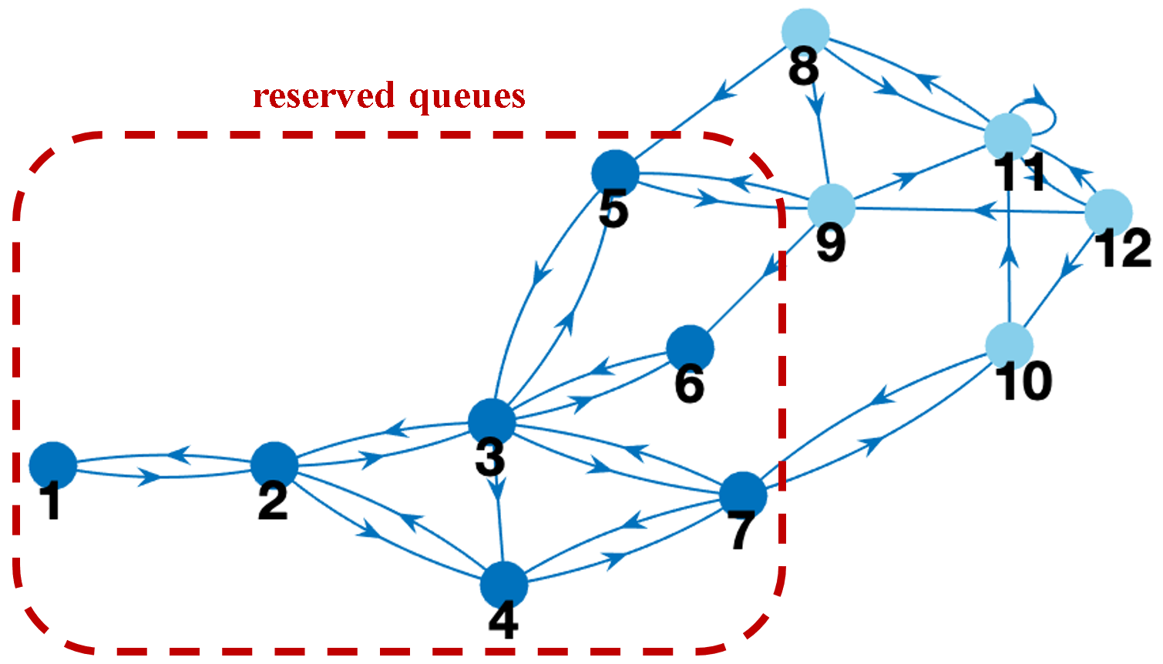}\vspace{0.2cm}
\end{minipage}
}
\hfill
\subfigure[The network graph $\Gamma(R)$ for the reduced model.]{ 
\begin{minipage}[c]{1\linewidth}
\centering    
\includegraphics[scale=0.2]{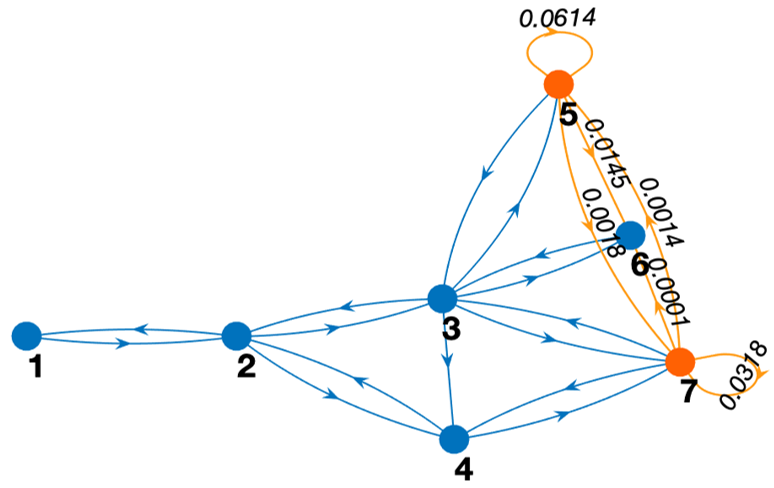}\vspace{0.2cm}
\end{minipage}}
\caption{Reduction of the $12$-queue open Jackson queueing network model to obtain a $7$ queue approximation is illustrated. }   
\label{fig8}    
\end{figure}

\begin{figure}[thpb]
\centerline{\includegraphics[scale=0.4]{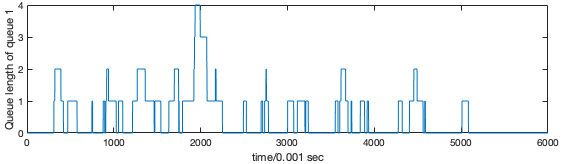}} 
\caption{The queue length of the target queue $1$ versus time.} 
\label{fig9}
\end{figure}

\begin{figure}[thpb]
\centering  
\subfigure[The mean queue length of the target queue 1.]{
\begin{minipage}[c]{1\linewidth}
\centering 
\includegraphics[scale=0.2]{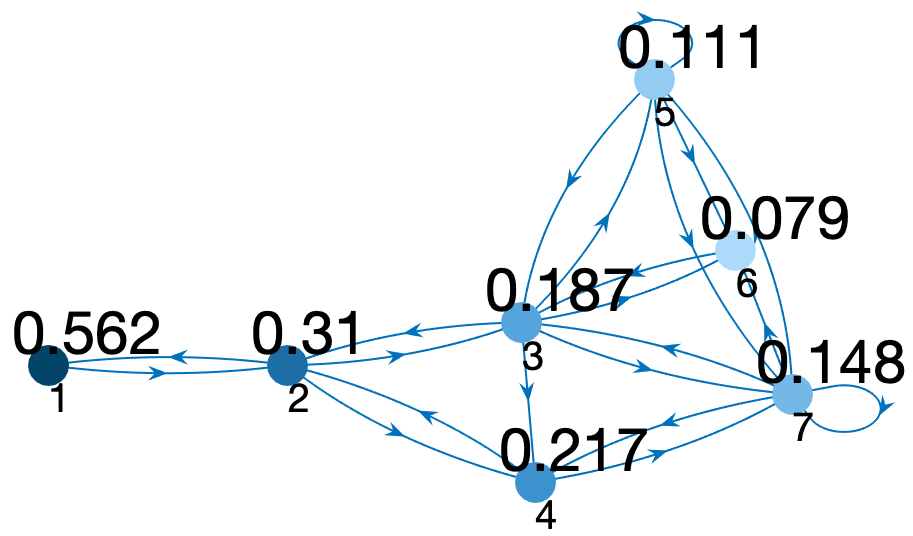}\vspace{0.2cm}
\end{minipage}
}
\hfill
\subfigure[The mean sojourn time of the target queue 1.]{ 
\begin{minipage}[c]{1\linewidth}
\centering    
\includegraphics[scale=0.2]{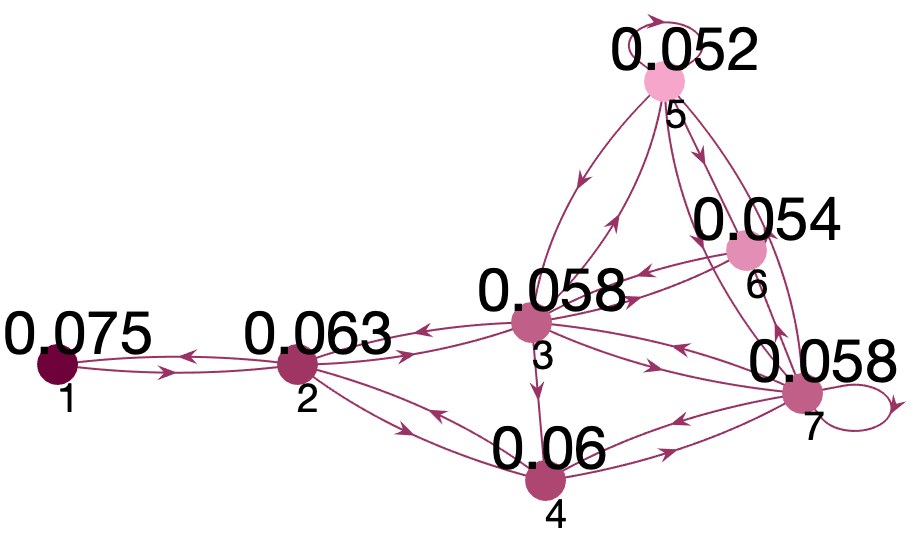}\vspace{0.2cm}
\end{minipage}}
\caption{Visualization of the performance metrics of a target queue in a reduction. The behavior of each receiver reserved from the original network are almost replicated in this structure-preserving reduction.}   
\label{fig10} 
\end{figure}

The analysis of mean performance metrics (mean queue length and sojourn time) is replicated for 
the reduced model, as shown in Figure \ref{fig10}.  
The formal analysis of the reduced model indicates that these first-order statistics should be exactly preserved in the reduced model as compared to the original.  Indeed, our computations of the statistics using long-duration simulations bear this out.

Finally, an additional performance metric -- the overall stay time of a job which arrives at each queue at each receiver location -- is compared for the original and the reduced model in Figure \ref{fig11}.  

\begin{figure}[thpb]
\centerline{\includegraphics[scale=0.38]{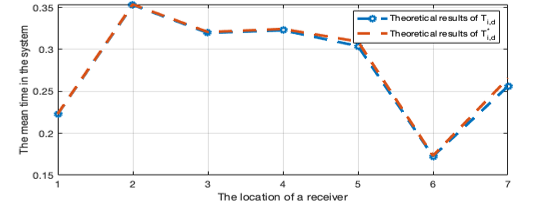}}
\caption{The overall stay time of jobs in the network for each receiver location, for the original and reduced model.}
\label{fig11}
\end{figure}

The stay-time metric is not guaranteed to be preserved in the reduced model.  However, as illustrated in the figure, the metric is found to be very well preserved via the reduction, with only a small error evident when the job enters at locations near the reduced subnetwork.

\bibliographystyle{IEEEtran}
\bibliography{IEEEabrv,ref.bib}

\end{document}